%% file: main.tex
\documentclass[review]{elsarticle}

\usepackage{lineno,hyperref}
\modulolinenumbers[5]
\usepackage{times}

\usepackage{soul}
\usepackage{url}
\usepackage{amsthm}
\usepackage[utf8]{inputenc}
\usepackage{graphicx}
\usepackage{amsmath,amssymb}
\usepackage{mathtools}
\usepackage{booktabs}
\usepackage{dsfont}
\allowdisplaybreaks
\newtheorem{thrm}{Theorem}

\newcommand{\ceil}[1]{\lceil #1 \rceil}

\usepackage[utf8]{inputenc}
\usepackage{nicefrac}
\usepackage{xparse}
\usepackage{xspace}
\usepackage{bm}
\usepackage{algorithm}
\usepackage[noend]{algpseudocode}
\usepackage{xcolor}
\usepackage{pifont}
\usepackage{url}
\usepackage{dsfont}

\usepackage{numdef} % for \c11
\num\def\can11{c_1^1}
\num\def\can21{c_2^1}
\num\def\cank1{c_k^1}

\num\def\can12{c_1^2}
\num\def\can22{c_2^2}
\num\def\cank-12{c_{k-1}^2}
\num\def\cank2{c_k^2}

\num\def\can13{c_1^3}
\num\def\can23{c_2^3}
\num\def\cank3{c_k^3}
\num\def\cank-13{c_{k-1}^3}

\num\def\can14{c_1^4}
\num\def\can24{c_2^4}
\num\def\cank4{c_k^4}
\num\def\cank-14{c_{k-1}^4}

\num\def\can1i1{c_1^i}
\num\def\can2i1{c_2^i}
\num\def\canki1{c_k^i}
\num\def\cank-1i1{c_{k-1}^i}

\num\def\can1k+1{c_1^{k+1}}
\num\def\can2k+1{c_2^{k+1}}
\num\def\cankk+1{c_k^{k+1}}
\num\def\cank-1k+1{c_{k-1}^{k+1}}

\num\def\can1k+2{c_1^{k+2}}
\num\def\can2k+2{c_2^{k+2}}
\num\def\cankk+2{c_k^{k+2}}
\num\def\cank-1k+2{c_{k-1}^{k+2}}
\input{macro.tex}

\begin{document}

\begin{frontmatter}
\title{Multi-Winner Election Control via Social Influence}

\author{Mohammad {Abouei Mehrizi}}
\ead{mohammad.aboueimehrizi@gssi.it}

\author{Gianlorenzo {D'Angelo}}
\ead{gianlorenzo.dangelo@gssi.it}

\address{Gran Sasso Science Institute (GSSI), L'Aquila, Italy.}

\input{trunk/abstract.tex}
\end{frontmatter}

\begin{keyword}
Computational Social Choice, Election Control, Multi-Winner Election, Social Influence.
\end{keyword}

%\linenumbers

% 

\input{trunk/abstract.tex}
\input{trunk/introduction.tex}
\input{trunk/multiWinnerelectionControl.tex}
\input{trunk/hardness.tex}

\input{trunk/SPVoting.tex}
\input{trunk/conclusions.tex}

%\clearpage
%\input{ref.bbl}
%\bibliographystyle{named}

\bibliography{references}

%\bibliographystyle{plain}

%\newpage
%\appendix
%\input{trunk/appendix}

%\input{trunk/MultiMessage}
\end{document}

%% file: macro.tex
\newcommand{\todo}[1]{\textcolor{blue}{\ding{46}{\sf}~\textbf{#1}}}
\newcommand{\MOV}{\ensuremath{\text{MoV}_{c}}\xspace}
\newcommand{\MOVD}{\ensuremath{\text{MoV}_{d}}\xspace}
\newcommand{\DOW}{\ensuremath{\text{DoW}_{c}}\xspace}
\newcommand{\DOWD}{\ensuremath{\text{DoW}_{d}}\xspace}
\newcommand{\CMEC}{{\sc{CMEC}}\xspace}

\newcommand{\DOV}{\ensuremath{\text{DoV}_{c}}\xspace}
\newcommand{\DOVD}{\ensuremath{\text{DoV}_{d}}\xspace}
\newcommand{\DOVP}{\ensuremath{\text{DoV}_{c}^+}\xspace}

\newcommand{\DOVN}{\ensuremath{\text{DoV}_{c}^-}\xspace}

\newcommand{\DMEC}{{\sc{DMEC}}\xspace}
\newcommand{\CDW}{{\sc{CDW}}\xspace}
\newcommand{\DDW}{{\sc{DDW}}\xspace}

\newcommand{\SPV}{{\sc{SPV}}\xspace}
\newcommand{\SPVF}{\tiny{spv}}

\newcommand{\Nout}[1]{N^{o}_{#1}}

\newcommand{\F}{\mathcal{F}}

\newcommand{\I}{\mathcal{I}}
\newcommand{\G}{\mathcal{G}}% set of all live-edge graphs
\newcommand{\E}{\mathds{E}} % Expected value

\newcommand{\Cmoa}[1]{C_{\text{\tiny{A}}}^{#1}}
\newcommand{\Cmob}{C_{\text{\tiny{B}}}}

\newcommand{\NP}{\mathit{NP}}
\newcommand{\cstar}{c_\star}

\newcommand{\R}{\mathbb{R}}
\NewDocumentCommand\Ex{gg}{%
	\ensuremath{%
		\mathbf{E} \IfNoValueTF{#1}{}{\left[ #1 \IfNoValueTF{#2}{}{\cond #2} \right]}
	}
}

\RenewDocumentCommand\Pr{gg}{%
	\ensuremath{%
		\mathbf{P} \IfNoValueTF{#1}{}{\left( #1 \IfNoValueTF{#2}{}{\cond #2} \right)}
	}
}

%% file: trunk/abstract.tex
\begin{abstract}
In an election, we are given a set of voters, each having a preference list over a set of candidates, that are distributed on a social network. We consider a scenario where voters may change their preference lists as a consequence of the messages received by their neighbors in a social network. Specifically, we consider a political campaign that spreads messages in a social network in support or against a given candidate and the spreading follows a dynamic model for information diffusion. When a message reaches a voter, this latter changes its preference list according to an update rule. The election control problem asks to find a bounded set of nodes to be the starter of a political campaign in support (constructive problem) or against (destructive problem) a given target candidate $c$, in such a way that the margin of victory of $c$ w.r.t. its most voted opponents is maximized. 

It has been shown that several variants of the problem can be solved within a constant factor approximation of the optimum, which shows that controlling elections by means of social networks is doable and constitutes a real problem for modern democracies. Most of the literature, however, focuses on the case of single-winner elections. 

In this paper, we define the election control problem in social networks for \emph{multi-winner elections} with the aim of modeling parliamentarian elections. Differently from the single-winner case, we show that the multi-winner election control problem is $\NP$-hard to approximate within any factor in both constructive and destructive cases. We then study a relaxation of the problem where votes are aggregated on the basis of parties (instead of single candidates), which is a variation of the so-called \emph{straight-party voting} used in some real parliamentarian elections. We show that the latter problem remains $\NP$-hard but can be approximated within a constant factor.
\footnote{Current article is an extended version of a conference paper presented in SIROCCO 2020~\cite{sirocco2020}.}
%\keywords{}
\end{abstract}
\begin{keyword}
    Computational Social Choice \sep Election Control \sep Influence Maximization \sep Social Influence.
\end{keyword}

%% file: trunk/introduction.tex
\section{Introduction}
Nowadays, social media are extensively used and have become a crucial part of our life.
Generating information and spreading in social media is one of the cheapest and most effective ways of advertising and sharing content and opinions. 
People feel free to share their opinion, information, news, or also gain something by learning or teaching in social media; 
on the other hand, they also use social media to get the latest news and information.
Many people even prefer to check social media rather than news websites.

Social media are also exploited during election campaigns to support some party or a specific candidate. Many political parties diffuse targeted messages in social media with the aim of convincing users to vote for their candidates. Usually, these messages are posted by influential users and diffused on the network following a cascade effect, also called social influence.
There are shreds of evidence of control election using the effect of social influence by spreading some pieces of information, including fake news or misinformation~\cite{pennycook2018prior}.
The presidential election of the United States of America is a real example. It has been shown that on average, ninety-two percent of Americans remembered pro-Trump false news, and twenty-three percent of them remembered the pro-Clinton fake news~\cite{allcott2017social}.
There are more real-life examples that have been presented in the literature~\cite{bond201261,ferrara2017disinformation,kreiss2016seizing,ElectionCampaigningSebastianStier2018}.

This motivated the study of election control problems in social networks by using dynamic models for influence diffusion. We are given a social network of voters, a set of candidates, and a dynamic model for diffusion of information that models the spread of messages produced by political campaigns. The problem asks to find a bounded set of voters/nodes to be the starter of a political campaign in support of a given target candidate $c$, in such a way that the margin of victory of $c$ w.r.t. its opponents is maximized. 
Each voter has its own preference list over the candidates and the winner of an election is determined by aggregating all preference lists according to some specific voting rule. Voters are autonomous, however their opinions about the candidates, and hence their preference lists, may change as a consequence of messages received by neighbors.
When a message generated by a political campaign reaches a node, this latter changes its preference list according to some specific update rule. When the campaign aims to make the target candidate win, we refer to the \emph{constructive} problem, while when the aim is to make $c$ lose, we refer to the \emph{destructive} problem.
This problem recently received some attention (see the next paragraph). Most of the works in the area, however, focus on single-winner voting systems, while several scenarios require voting systems with multiple winners, e.g., parliamentarian elections. 

%What is the problem?
In this paper, we consider the problem of \textit{multi-winner election control via social influence}, where there are some parties, each with multiple candidates, and we want to find at most $B$ nodes to spread a piece of news in the social network in such a way that a target party elects a large number of its candidates.
In this model, more than one candidate will be elected as the winner, and parties try to maximize some function of the number of winners from their party.
We considered this problem for some well-known objective functions in both constructive and destructive cases.

\paragraph{Our results} 
We introduce the multi-winner election control problem via social influence and show that it is $\NP$-hard to approximate within any factor $\alpha > 0$, for two common objective functions known as \textit{margin of victory} and \textit{difference of winners} using a general scoring rule. This is in contrast with the previous work on single-winner election control through IM, in which it is possible to approximate the optimum within a constant factor.
The hardness results hold for both constructive and destructive cases. %, and both plurality and $k$-approval scoring rules.
Given the hardness result, we focus on a relaxed version of the problem, which is a variation of straight-party voting. We show that this latter remains $\NP$-hard but admits a constant factor approximation algorithm for both constructive and destructive cases.

%*Outline of the paper?
\paragraph{Outline} 
The rest of the paper is managed as follows.
In Section~\ref{sec:relatedWorks}, we consider the articles that make the state of the art of the problem.
Section~\ref{sec:multiWinnerElectionControl} contains our model, formulation, and objective functions.
We show that the problem is $\NP$-hard to approximate within any approximation factor $\alpha > 0$ in Section~\ref{sec:hardness}. 
Section~\ref{sec:straightParty} presents a relaxed version of the problem, a variation of straight-party voting, where the voters vote for a party instead of candidates.
Finally, in Section~\ref{sec:conclusion}, the summary of the results and future works are discussed.

%*State of the art?
%\mohammad{Write the state of the art here.}
%\paragraph{Related Work}
\section{Related Works}
\label{sec:relatedWorks}
There is an extensive literature about manipulation or control of elections. We refer to the survey in~\cite{faliszewski2016control} for relevant work on election control without using social networks.
In the following, we focus on election control problems where the voters are the nodes of a social network, which recently received some attention. 
%They make use of models for influence spread to define optimization problems where the objective is to make some target candidate win. 

Finding strategies to maximize the spread of influence in a network is one of the main topics in network analysis. Given a network and a dynamic model for the diffusion of influence, find a bounded set of nodes to be the starters of a dynamic process of influence spread in such a way that the number of eventually influenced users is maximized. The problem, known as Influence Maximization (IM), has been introduced by Domingos and Richardson~\cite{DomingosR01,RichardsonD02} and formalized by Kempe et al., who gave a $(1-1/e)$-approximation algorithm~\cite{kempe2015maximizing} for two of the most used dynamic models, namely Independent Cascade Model (ICM) and Linear Threshold Model (LTM). We point the reader to the book by Chen et al.~\cite{ChenCL13} and to~\cite{kempe2015maximizing}.

Wilder and Vorobeychik~\cite{wilder2018controlling} started the study of election control by means of IM. They defined an optimization problem that combines IM and election control called \emph{election control through influence maximization} that is defined as follows. We are given a set of candidates, a social network of voters, each having a preference list over the candidates, a budget $B$, and a specific target candidate $\cstar$. The network allows the diffusing influence of individuals according to  ICM. When a node/voter $v$ is influenced, it changes its preference list in such a way that the rank of $\cstar$ in the preference list of $v$ is promoted (constructive) or demoted (destructive) by one position. At the end of a diffusion process, the voters elect a candidate according to the plurality rule~\cite{zwicker2016voting}.
The problem asks to find a set of at most $B$ nodes to start a diffusion process in such a way that the chances for $\cstar$ to win (constructive) or lose (destructive) at the end of the diffusion are maximized. Wilder and Vorobeychik used the Margin of Victory (MoV) as an objective function and showed that there exists a greedy algorithm that approximates an optimal solution by a factor $1/3(1-1/e)$ for constructive and $1/2(1-1/e)$ for the destructive case.
The same problem has been extended to LTM and general scoring rules~\cite{zwicker2016voting} by Cor\`o et al.~\cite{IJCAI-19,aamas19}. They have shown that the problem can be approximated within the same bound.
A similar problem has been studied in~\cite{faliszewski2018opinion}. The authors consider a network where each node is a set of voters with the same preference list, and edges connect nodes whose preference lists differ by the ordering of a single pair of adjacent candidates. They use a variant of LTM for influence diffusion and show that the problem of making a specific candidate win is $\NP$-hard and fixed-parameter tractable w.r.t. the number of candidates.
% 
% Bredereck et al.\ studied 
Bredereck and Elkind~\cite{bredereck2017manipulating} considered the following election control problem. Given a network where the influence spread according to a variant of LTM in which each node has a fixed threshold, and all edges have the same weight, find an optimal way of bribing nodes or add/delete edges in order to make the majority of nodes to vote for a target candidate.
A different line of research investigates a model in which each voter is associated with a preference list over the candidates, and it updates its list according to the majority of opinions of its neighbors in the network~\cite{auletta2015minority,botan2017propositionwise,brill2016pairwise}. 
All the previous works on election control through IM consider single-winner voting systems.
Multi-winner voting systems raised recent and challenging research trends, we refer to a recent book chapter~\cite{Faliszewski2017multiwinner} and references therein.

%% file: trunk/multiWinnerelectionControl.tex
\section{Multi-winner Election Control}
\label{sec:multiWinnerElectionControl}
In this section, we introduce the multi-winner election control problem. We consider elections with $k$ winners and general scoring rule as a voting system, which includes many well-known scoring rules, such as plurality, approval, Borda, and veto~\cite{zwicker2016voting}. We first introduce the models that we use for diffusion of influence and for updating the preference list of voters. Then we introduce the objective functions for the election control problem in both constructive and destructive cases.

%\todo{Cite the other papers somewhere around here.}
%Voters will cast a vote for their most preferred candidate, and the first $k$ candidates with the most score will be elected.
\paragraph{Model for influence diffusion}
%\textit{Influence maximization} is the problem of finding a small subset of nodes (as seed nodes) in a social network to maximize the spread of influence.
We use the \textit{Independent Cascade Model} (ICM) for influence diffusion~\cite{kempe2015maximizing}.
In this model, we are given a directed graph $G=(V,E)$, where each edge $(u, v) \in E$ has a weight $b_{uv} \in [0, 1]$.
The influence starts with a set of seed nodes $S$ and keeps activating the nodes in at most $|V|$ discrete steps.
In the first step, all the seed nodes $S$ become active. In the next steps $i>1$, all the nodes that were active in step $i-1$ remain active, moreover, each node $u$ that became active at step $i-1$ tries to activate its outgoing neighbors at step $i$ with probability $b_{uv}$, for each node  $v \in \Nout{u}$.
%When node $u$ becomes active at step $i-1$, it activates each of its outgoing neighbors $v \in \Nout{u}$ with probability $b_{uv}$.
An active node will try to activate its outgoing neighbors independently and only once.
The process stops when no new node becomes active. We denote by $A_S$ the set of nodes that are eventually active by the diffusion started by the seeds $S$.

%\subsection{Model and Problem Formulation}
\paragraph{Model for multi-winner election control}
\label{subsec:modelAndFormulation}
We consider a multi-winner election in which $k$ candidates will be elected. 
Let $G = (V, E)$ be a directed social graph, where the nodes are the voters in the election, and the edges represent social relationships among users. The voters influence each other the same as ICM.
We consider $t$ parties $C_1,C_2,\ldots,C_t$, each having $k$ candidates, $C_i = \{c_1^i, c_2^i, \dots, c_k^i \}$, $1 \leq i \leq t$. 
Without loss of generality, we assume that $C_1$ is the target party. The set of all candidates is denoted by $C$, i.e. $C=\bigcup_{i=1}^t C_i$.
Each voter $v \in V$ has a preference list $\pi_v$ over the candidates. For each $c\in C$, we denote by $\pi_v(c) \in \{1, 2, \dots, tk \}$ the rank (or position) of the candidate $c$ in the preference list of node $v$.

Given a budget $B$, we want to select a set of $B$ seed nodes that maximizes the number of candidates in $C_1$ who win the election after a political campaign that spread according to ICM starting from nodes $S$ (see next section for a formal definition of the objective functions).
%
% Let define $\Nin{v}$ and $\Nout{v}$ the set of incoming and outgoing neighbors of node $v \in V$, respectively.
%We use the ICM diffusion model, where $A_S$ is defined as a set of nodes that are eventually active after the diffusion process started from a set of seed nodes $S$.
\footnote{In the remainder of the paper, by \textit{after $S$}, we mean \textit{after the diffusion process started from the set of seed nodes $S$.}}

After $S$, nodes in $A_S$ will change the positions of candidates in $C_1$ in their preferences list. In contrast, nodes not in $A_S$ will maintain their original preference list.
The update rule for active nodes depends on the position of the target candidates and the goal of the campaign, i.e., if it is a constructive or a destructive one.
We denote the preference list of node $v$ after the process by $\tilde{\pi}_v$. If $v\not\in A_S$, then $\tilde{\pi}_v = \pi_v$. In the following, we focus on nodes $v\in A_S$.

In the constructive case, like in the model in~\cite{wilder2018controlling}, the position of the target candidates in the list of active nodes will be decreased by one, if there is at least one opponent candidate in a smaller rank.
The candidates who are overtaken will be demoted by the number of target candidates that were just after them.

Formally, in the constructive case, the position of the candidates after the diffusion starting from seed $S$ will change as follows. For each node $v\in A_S$ and for each target candidate $c \in C_1$, the new position of $c$ in $v$ is
\[
\tilde{\pi}_v(c)\! =\! \left\{\!
\begin{array}{ll}
    \pi_v(c)-1 &  \mbox{if }  \exists~c' \in C \setminus C_1 \mbox{ s.t. } \pi_v(c')<\pi_v(c)\\
     \pi_v(c)& \mbox{otherwise,}
\end{array}
\right.
\]
while, for each opponent candidate $c \in C \setminus C_1$, if there exists a candidate $c' \in C_1$ s.t. $\pi_v(c') = \pi_v(c) + 1$, then we set the new position of candidate $c$ in preference list of node $v$ as 
\begin{multline*}
    \tilde{\pi}_v(c) = \pi_v(c) + \\ |\{c'' \in C_1 ~|~ \pi_v(c'')>\pi_v(c) \land
    (\nexists ~\bar{c} \in C \setminus C_1: \pi_v(c) < \pi_v(\bar{c}) < \pi_v(c''))\}|~,
\end{multline*}
otherwise, we set $\tilde{\pi}_v(c) = \pi_v(c)$.

In the destructive case, we want to reduce the number of winners in $C_1$ and then each node $v \in A_S$ increases their position by one, if it is possible.
Formally, after $S$ the preferences list of the candidates will change as follows.
For each node $v\in A_S$ and for each target candidate $c \in C_1$, the new position of $c$ in $v$ is
\[
\tilde{\pi}_v(c)\! =\! \left\{\!
\begin{array}{ll}
    \pi_v(c)+1 &  \mbox{if }  \exists~c' \in C \setminus C_1 \mbox{ s.t. } \pi_v(c')>\pi_v(c)\\
     \pi_v(c)& \mbox{otherwise,}
\end{array}
\right.
\]
while for $c \in C \setminus C_1$, if there exists a candidate $c' \in C_1$ s.t. $\pi_v(c') = \pi_v(c) - 1$ we have
\begin{multline*}
    \tilde{\pi}_v(c) = \pi_v(c) - \\  |\{c'' \in C_1 ~|~ \pi_v(c'')<\pi_v(c) \land (\nexists~\bar{c} \in C \setminus C_1: \pi_v(c'') < \pi_v(\bar{c}) < \pi_v(c))\}|~,
\end{multline*}
otherwise, we set $\tilde{\pi}_v(c) =\pi_v(c)$.

%\todo{Change the explanation of this example.}
As an example, if there are two parties with three candidates each, and the initial preferences list of a node is $(c_1^2, c_1^1, c_2^1, c_2^2, c_3^1, c_3^2)$, then if the node becomes active its preferences list in the constructive case will be $(c_1^1, c_2^1, c_1^2, c_3^1, c_2^2, c_3^2)$, i.e., all of the candidates $c_i^1$ will promote, and all the overtaken candidates will demote;
while in the destructive case, it will be $(c_1^2, c_2^2, c_1^1, c_2^1, c_3^2, c_3^1)$, and all of the candidates in our target party demote, and all the overtaken candidates will promote.

The above rule for updating the preference lists is commonly used in the literature~\cite{IJCAI-19,wilder2018controlling}.
In this model, we consider just one message, which contains some positive/negative information about the target party that will affect all the target candidates.

We consider a non-increasing scoring function $f(i)$, $1 \leq i \leq |C|$, such that for all $j > i > 0$ we have $f(j) \leq f(i)$.
A candidate $c \in C$ gets $f(\pi_v(c))$ and $f(\tilde{\pi}_v(c))$ points from voter $v$ before and after a diffusion, respectively.
In other words, each voter will reveal his preferences list, and each candidate will get some score according to his position in the list and the scoring function.
Also, we assume w.l.o.g. that there exist $1 \leq i < j \leq |C|$ such that $f(i) > f(j)$, i.e., the function does not return a fixed number for all ranks. The score of a candidate $c$ is the sum of the scores received by all voters. The $k$ candidates with the highest score will be elected.

We denote by $\F(c, \emptyset), \F(c, S)$, respectively, the expected overall score received by candidate $c$ before and after $S$; formally, $\forall c \in C:$
\begin{align*}
\F(c, \emptyset) = \sum_{v \in V} f(\pi_v(c)), \\ 
\F(c, S) = \E_{A_S} \Big[\sum_{v \in V} f(\tilde{\pi}_v(c))\Big].
\end{align*}

%\subsection{Objective Functions}
\paragraph{Objective Functions}
\label{subsec:objectiveFunction}
The objective function for the constructive election control problem in the single-winner case is maximizing the \emph{margin of victory} (MoV) defined in~\cite{wilder2018controlling}.
Let us consider the difference between the votes for the target candidate and those for the most voted opponent candidate. 
MoV is the change of this value after $S$. Note that the most voted opponent before and after $S$ might change.
The notion of MoV captures the goal of a candidate to have the largest margin in terms of votes w.r.t. any other candidate.
We extend the above definition of MoV in the case of  \emph{multi-winner} election control. 
Since the main goal is to elect more candidates from the target party, then we define the objective function in terms of the number of winning candidates in our target party before and after $S$.

Given a set $A_S$ of nodes that are active at the end of a diffusion process started from $S$, we denote by $\F_{A_S} (c)$ the score that a candidate $c \in C$ receives if the activated nodes are $A_S$, and by $\mathcal{Y}_{A_S}(c)$ the number of candidates that have less score than the candidate $c$. As a tie-breaking rule, we assume that $c_i^j$ has priority over $c_{i'}^{j'}$ if $j<j'$, or $j=j'$ and $i<i'$. In particular, the target candidates have priority over opponents when they have the same score.
Then, for each $c_i^j\in C$, $i\in\{1,\ldots,k\}$, $j\in\{1,\ldots,t\}$, $\mathcal{Y}_{A_S}(c_i^j, S)$ is defined as
\begin{multline*}
    \mathcal{Y}_{A_S}(c_i^j) = \Big|\{c_{i'}^{j'} \in C~|~ \F_{A_S}(c_i^j) > \F_{A_S}(c_{i'}^{j'}) \lor \\
    (\F_{A_S}(c_i^j) = \F_{A_S}(c_{i'}^{j'}) \land (j<j' \lor (j=j' \land i<i'))\}\Big|.
\end{multline*}  
%Similarly, $\F(c, \emptyset) = \sum_{v \in V}  f(\pi_v(c))$ is the score received by candidate $c$ before any diffusion.

For a party $C_i$, we define $\F(C_i, S)$ as the expected number of candidates in $C_i$ that win the election after $S$; formally,
%More formally, 
\begin{equation}\label{def:scoreparty}%\textstyle
\F(C_i, S) = \mathbb{E}_{A_S}\left[ \sum_{c \in C_i}\mathds{1}_{\mathcal{Y}_{A_S}(c) \geq (t-1)k}\right].
\end{equation}

We denote by $\Cmob$ and $\Cmoa{S}$ the opponent party with the highest number of winners \emph{before} and \emph{after} $S$, respectively.
For the constructive case, the  \textit{margin of victory} ($\MOV$) for party $C_1$, w.r.t. seeds $S$, is defined as follows:
\begin{align*}
%\label{def:MOVPartiesDiff}
    \MOV(C_1, S) = \F(C_1, S) - \F(\Cmoa{S}, S) - \big( \F(C_1, \emptyset) - \F(\Cmob, \emptyset) \big),
\end{align*}
while for the destructive case, it is defined as:
 \begin{align*}
%\label{def:MOVDestPartiesDiff}
    \MOVD(C_1, S) = \F(C_1, \emptyset) - \F(\Cmob, \emptyset) - \big( \F(C_1, S) - \F(\Cmoa{S}, S) \big).
\end{align*}
The \emph{Constructive (Destructive, resp.) Multi-winner Election Control problem} (\CMEC (\DMEC, resp.)) asks to find a set $S$ of $B$ seed nodes that maximizes $\MOV(C_1, S)$ ($\MOVD(C_1, S)$, resp.), where $B\in \mathbb{N}$ is a given budget.

In some scenarios, it is enough to maximize the difference between the number of our target candidates who win the election before and after $S$; we call this objective function the \textit{difference of winners} (DoW), and for the constructive case we define it as: 
\begin{align*}
%\label{def:DOWConstructive}
    \DOW(C_1, S) = \F(C_1, S) - \F(C_1, \emptyset).
\end{align*}
While for the destructive model it is defined as:
\begin{align*}
%\label{def:DOWDestructive}
    \DOWD(C_1, S) = \F(C_1, \emptyset) - \F(C_1, S).
\end{align*}
The problems of finding a set of at most $B$ seed nodes that maximize \DOW and \DOWD, for a given integer $B$, are called 
\emph{Constructive Difference of Winners} (\CDW) and \emph{Destructive Difference of Winners}  (\DDW), respectively.

%% file: trunk/hardness.tex
\section{Hardness Results}
\label{sec:hardness}
In this section, we show the hardness of approximation results for the problems defined in the previous section.
%As we saw in previous sections, we focus on \textit{plurality} scoring rule. 
%In Section~\ref{subsec:hardnessConstructive}, w
We first focus on the constructive case and prove that \CMEC and \CDW are $\NP$-hard to approximate within any approximation factor $\alpha > 0$. Then,
%In Section~\ref{subsec:destructivePlurality}, 
we show that the same results hold for \DMEC and \DDW. 
All the results hold even when the instance is deterministic (i.e. $b_{uv}=1$, for each $(u,v)\in E$) and when $t=3$ and $k=2$. Note that for $t=1$ the problem is trivial and for $k=1$ the problem reduces to the single-winner case.

\input{trunk/hardness/constructive.tex}

\input{trunk/hardness/destructive.tex}

%% file: trunk/hardness/constructive.tex
%\subsection{Constructive Election Control}
%\label{subsec:hardnessConstructive}
\paragraph{Constructive Election Control}
%We now focus on the constructive problems.
We first give an intuition of the hardness of approximation proof, which is formally given in Theorem~\ref{theorem:alphaAppNpHardness}. 
Consider an instance of the constructive case in which $t=k=2$, $C_1=\{c_1^1,c_2^1\}$, $C_2=\{c_1^2,c_2^2\}$, and $C = C_1 \cup C_2$. The weight of all edges are equal to 1, that is, the diffusion is a deterministic process. 
Also, assume the scoring rule is plurality, i.e., $f(1) = 1, f(2) = f(3) = f(4) = 0$.
Moreover, the nodes are partitioned into two sets of equal size, $V_1$ and $V_2$. 
In the preferences lists of all nodes in $V_1$, candidate $c_1^2$ is in the first position and $c_1^1$ is in the second position, while in the preferences lists of nodes in $V_2$, candidate $c_2^2$ is in first position and $c_2^1$ is in second position. 
In this instance, initially party $C_1$ does not have any elected candidate, that is, $\F(c_1^1,\emptyset) = \F(c_2^1,\emptyset) = 0$, $\F(c_1^2,\emptyset) = |V_1|$, $\F(c_2^2,\emptyset) = |V_2|$, $\F(C_1,\emptyset) = 0$, and $\F(C_2,\emptyset) = 2$.

Consider a diffusion process starting from seeds $S$ that activate nodes $A_S$ (note that, since the weights are all equal to 1, $A_S$ is a deterministic set for any fixed $S$). The number of candidates that receives fewer votes than a target candidate $c_i^1$ after the diffusion process is
\begin{multline*}
    \mathcal{Y}_{A_S}(c_i^1) = \Big|\{c_{i'}^{j'} \in C~|~ \F_{A_S}(c_i^1) > \F_{A_S}(c_{i'}^{j'}) \lor \\
    (\F_{A_S}(c_i^1) = \F_{A_S}(c_{i'}^{j'}) \land (j'=2 \lor i<i'))\}\Big|.
\end{multline*}

Let us consider the case $i=1$ and analyze the conditions that a seed set $S$ must satisfy in order to include a candidate in the above set, i.e., make %one of
$c_1^1$ win. 
We analyze the three other candidates $c_{i'}^{j'}$ separately.
%, in order to satisfy the above condition concerning another candidate $c_{i'}^{j'}$, we need to find a seed sets such that
\begin{itemize}
    \item If $j'=2$ and $i'=1$, then we must have $\F_{A_S}(c_1^1) \geq \F_{A_S}(c_1^2)$. Since the preferences list of each active nodes in $V_1$ is updated in a way that $c_1^1$ moves to the first position and $c_1^2$ moves to the second position, and the active nodes in $V_2$ do not affect the rankings of $c_1^1$ and $c_1^2$, we have that $\F_{A_S}(c_1^1) = |A_S\cap V_1|$ and  $\F_{A_S}(c_1^2)=|V_1\setminus A_S|$. Therefore, $\F_{A_S}(c_1^1) \geq \F_{A_S}(c_1^2)$ if and only if $|A_S\cap V_1| \geq |V_1\setminus A_S| = |V_1|-|V_1\cap A_S|$, which means that $|A_S\cap V_1| \geq |V_1|/2$.
    \item If $j'=2$ and $i'=2$, then we must have $\F_{A_S}(c_1^1) \geq \F_{A_S}(c_2^2)$. In this case, we still have $\F_{A_S}(c_1^1) = |A_S\cap V_1|$, and, since $c_2^2$ is moved down by one position for each active node in $V_2$, then $\F_{A_S}(c_2^2) = |V_2\setminus A_S|$. This implies that $\F_{A_S}(c_1^1) \geq \F_{A_S}(c_2^2)$ if and only if $|A_S\cap V_1| \geq |V_2\setminus A_S| = |V_2|-|V_2\cap A_S|$, which means $|A_S\cap V_1| + |A_S\cap V_2|\geq |V_2|$.
    \item If $j'=1$ and $i'=2$, then we must have $\F_{A_S}(c_1^1) \geq \F_{A_S}(c_2^1)$. We again have $\F_{A_S}(c_1^1) = |A_S\cap V_1|$, and, since $c_2^1$ is moved by one position up for each active node in $V_2$, then $\F_{A_S}(c_2^1) = |A_S\cap V_2|$. Therefore,  $\F_{A_S}(c_1^1) \geq \F_{A_S}(c_2^1)$ if and only if $|A_S\cap V_1| \geq|A_S\cap V_2|$.
\end{itemize}
Similar conditions hold for $i=2$. 

In order to elect candidate $c_1^1$ we should have $\mathcal{Y}_{A_S}(c_1^1)\geq (t-1)k=2$, which means, we should find a seed set that satisfies at least two of the above conditions (or the corresponding conditions to elect $c_2^1$). Note that finding a seed set $S$ of size at most $B$ that satisfies any pair of the above conditions is a $\NP$-hard problem since it requires to solve the IM problem, which is $\NP$-hard even when the weight of all edges is 1~\cite{kempe2015maximizing}.
Let us assume that an optimal solution is able to elect both candidates in $C_1$ (e.g. by influencing $|V_1|/2$ nodes from $V_1$, and $|V_2|/2$ nodes from $V_2$), then the optimal \MOV and \DOW are equal to $4$ and $2$, respectively.
Moreover, in this case $\Cmoa{S} = \Cmob = C_2$, then $\MOV(C_1, S) = \F(C_1, S) - \F(C_1, \emptyset) + \F(C_2, \emptyset) - \F(C_2, S)$. Since $\F(C_2, \emptyset) - \F(C_2, S) = \F(C_1, S)- \F(C_1, \emptyset)$, i.e., for each candidate lost by $C_2$ there is a candidate gained by $C_1$, then $\MOV(C_1, S) = 2 (\F(C_1, S)- \F(C_1, \emptyset)) = 2\DOW(C_1,S)$. Since $\F(C_1,\emptyset) = 0$, any approximation algorithm for \CDW or \CMEC must find a seed set $S$ s.t. $\F(C_1,S) >0$ and this requires to elect at least one candidate in $C_1$ (see Equation~\eqref{def:scoreparty}), which is $\NP$-hard. 
It follows that it is  $\NP$-hard to approximate \CMEC and \CDW within any factor, as formally shown in the next theorem.

\begin{thrm}
\label{theorem:alphaAppNpHardness}
    It is $\NP$-hard to approximate \CMEC and \CDW within any factor $\alpha > 0$. 
\end{thrm}  

\begin{proof}
We reduce the decision version of the deterministic IM problem, to \CMEC and \CDW, where \textit{deterministic} refers to the weight of the edges, which is equal to 1. 
Let us define the decision version of the IM problem as follows:
\emph{
Given a directed graph $G=(V, E)$ and budget $B\leq |V|$. Is there a set of seed nodes $S \subseteq V$ such that $|S| \leq B$ and $A_S = V$?
}

Let $\mathcal{I}(G, B)$ be a deterministic instance of the decision IM problem (then, using a given seed set $S$, we can find the exact number of activated nodes in polynomial time). 
We create an instance $\mathcal{I}'(G', B)$ of \CMEC and \CDW, where $G'=(V \cup V', E \cup E')$. We use the same budget $B$ for both problems.
%First, 
We first investigate the case where $t = 3, k = 2$%., we then generalize to any $t, k > 2$. In the end, we generalize the proof for any $t, k > 2$.
, and consider two different cases as follows.
\begin{description}
    \item[C1.] If $f(1) = f(2) = f(3) = a, f(4) = f(5) = f(6) = b$ for $a, b \in \R \land a > b \geq 0$, we call this case \emph{exceptional}, and do as follows.
    
   \begin{itemize}
        \item 
        For each $v \in V$ we add one more node in $V'$ and it has just one incoming edge from $v$, i.e., $\forall v \in V: v_1 \in V', (v, v_1) \in E'$. 

        \item
        We set the preferences of all nodes $v \in V$ and its new outgoing neighbor as follows:
        \begin{align*}
            v : \can12 \succ \can22 \succ \can13 \succ \can11 \succ \can21 \succ \can23, \\
            v_1 : \can13 \succ \can23 \succ \can12 \succ \can21 \succ \can11 \succ \can22,
        \end{align*}
        where $c \succ c'$ means $c$ is preferred to $c'$.
   \end{itemize}
    
    \item[C2.] For any non-increasing scoring function except the exceptional ones, we call it general and do as follows.
    \begin{itemize}
   
        \item 
        For each $v \in V$ we add three more nodes in $V'$ and each of them has just one incoming edge from $v$; formally, 
        \begin{align*}
            \forall v \in V: v_1, v_2, v_3 \in V', \\ (v, v_1), (v, v_2), (v, v_3) \in E'.
        \end{align*}

        \item 
        We set the preferences of all nodes $v \in V$ and its new outgoing neighbors as follows:
        \begin{align*}
            v &: \can12 \succ \can11 \succ \can13 \succ \can22 \succ \can21 \succ \can23,\\
            v_1 &: \can22 \succ \can21 \succ \can23 \succ \can12 \succ \can11 \succ \can13,\\
            v_2 &: \can12 \succ \can13 \succ \can11 \succ \can22 \succ \can23 \succ \can21,\\
            v_3 &: \can22 \succ \can23 \succ \can21 \succ \can12 \succ \can13 \succ \can11.
        \end{align*}

    \end{itemize}

\end{description}

In both cases, the weight of all edges is 1, i.e., $b_{uv} = 1$ for all $(u, v) \in E \cup E'$.

The score of candidates before any diffusion is as follows.
\begin{description}
    \item[C1.]
    \begin{align*}
        \F(\can11, \emptyset) &= \F(\can21, \emptyset) = |V|(f(4) + f(5)) = 2b|V|, \\
        \F(\can12, \emptyset) &= \F(\can13, \emptyset) = |V|(f(1) + f(3)) = 2a|V|, \\
        \F(\can22, \emptyset) &= \F(\can23, \emptyset) = |V|(f(2) + f(6)) = (a + b) |V|.
    \end{align*}
    Since $a > b \geq 0$, it yields that $\F(C_2, \emptyset) = \F(C_3, \emptyset) = 1$, $\F(C_1, \emptyset)  = 0$, and none of our target candidates win the election.
    
    \item[C2.]
    \begin{align*}
        \F(\can11, \emptyset) &= \F(\can21, \emptyset) = |V|(f(2) + f(3) + f(5) + f(6)), \\
        \F(\can12, \emptyset) &= \F(\can22, \emptyset) = |V|(2f(1) + 2f(4)), \\
        \F(\can13, \emptyset) &= \F(\can23, \emptyset) = |V|(f(2) + f(3) + f(5) + f(6)).
    \end{align*}
    Since $f(\cdot)$ is a non-increasing function, it yields $\F(C_2, \emptyset) = 2$ and $\F(C_1, \emptyset) = \F(C_3, \emptyset) = 0$ and none of our target candidates win the election.

\end{description}

In $\mathcal{I}'(G', B)$, in both cases, 
%$\mathcal{I}'(G', B)$, concerning both cases,
all of the nodes $v \in V \cup V'$ become active if and only if all of the nodes $v \in V$ become active. Indeed, by definition, if $V \subseteq A_S$, then for each node $u \in V'$ there exists an incoming neighbor $v \in V$ s.t. $(v, u) \in E'$ and $b_{vu} = 1$, then if $v$ is active, also $u$ becomes active.

Suppose there exists an $\alpha-$approximation algorithm called \textit{$\alpha$-appAlg} for \CDW (resp. \CMEC) and it returns $S \subseteq V \cup V'$ as a solution. 
We show that, by using the seed nodes $S$ returned by the algorithm \textit{$\alpha$-appAlg}, we can find the answer for the decision IM problem.
We will show that $\DOW(C_1, S) > 0$ (resp. $\MOV(C_1, S) > 0$) if and only if $S$ activates all of the nodes, i.e.,  $A_S = V \cup V'$.
%, where  $A_S$ is the set of activated nodes after $S$. 
That is  $\DOW(C_1, S) > 0$ (resp. $\MOV(C_1, S) > 0$) if and only if the answer to the decision IM problem is YES.

W.l.o.g., we assume $S \subseteq V$, because if there exists a node $u \in S \cap V'$, we can replace it with the node $v \in V \mbox{ s.t. } (v, u) \in E'$. Since $b_{uv}=1$, this does not decrease the value of $\DOW(C_1,S)$ or $\MOV(C_1,S)$.

We now show that if $\DOW(C_1, S) > 0$ (resp. $\MOV(C_1, S) > 0$),  then $A_S = V \cup V'$. By contradiction, assume that $S$ will not activate all of the nodes, i.e., there exists a node $v$ in $V \setminus A_S$.
Then, the score of the candidates will be as follows.
\begin{description}
    \item[C1.]
    \begin{align*}
        \F(\can11, S) &%= \F(\can21, S) = |V|(f(4) + f(5)) 
        \leq 
        (a+b)(|V|-1) + 2b, \\
        \F(\can12, S) &%= \F(\can13, S) = |V|(f(1) + f(3)) 
        \geq 
        (a+b)(|V|-1) + 2a, \\
        \F(\can22, S) &= \F(\can23, S) = |V|(f(2) + f(6)) =
        (a + b) |V|.
    \end{align*}
    Since $a > b \geq 0$, then none of the target candidates will be among the winners, i.e., $\F(C_2, S) = \F(C_3, S) = 1$ and $\F(C_1, S)  = 0$ and $\DOW(C_1, S) = \MOV(C_1, S) = 0$.
    
    \item[C2.]
    \begin{align*}
        \F(\can11, S) &= \F(\can21, S) \leq 
        (|V|-1)(f(1) + f(2) + f(4) + f(5)) + \\ &\quad 
        (f(2) + f(3) + f(5) + f(6)), \\
        \F(\can12, S) &= \F(\can22, S) \geq 
        (|V|-1)(f(1) + f(2) + f(4) + f(5)) + \\ &\quad 
        (2f(1) + 2f(4)), \\
        \F(\can13, S) &= \F(\can23, S) \geq 
        (|V|-1)(f(3) + f(4) + 2f(6)) + \\ &\quad 
        (f(2) + f(3) + f(5) + f(6)).
    \end{align*}
    Since $f(\cdot)$ is a non-increasing function, then $\F(C_1,S) = \F(C_3,S) = 0$ and $\F(C_2,S) = 2$. Therefore $\DOW(C_1, S) = \MOV(C_1, S) = 0$.
\end{description}
In both cases we have a contradiction.
To show the other direction, if all of the nodes become active, then the score of candidates will be as follows.
\begin{description}
    \item[C1.] For each $ 1 \leq i \leq 2$ and $1 \leq j \leq 3, \F(c_i^j, S) = (a + b) |V|$.
    Due to the tie-breaking rule %(recall that, by assumption, our target candidates will be preferred in the tie condition), 
    it follows that both of our target candidates will be among the winners, i.e., $\F(C_1, S)  = 2$ and $\F(C_2, S) = \F(C_3, S) = 0$.
    
    \item[C2.] 
    \begin{align*}
        \F(\can11, S) &= \F(\can21, S) = \F(\can12, S) = \F(\can22, S) =  \\ &\quad 
        |V|(f(1) + f(2) + f(4) + f(5)), \\
        \F(\can13, S) &= \F(\can23, S) = |V|(f(3) + f(4) + 2f(6)).
    \end{align*}
    Then, $\F(C_1,S) = 2$ and $\F(C_2,S) = \F(C_3,S) = 0$. 
\end{description}

Therefore we have $\DOW(C_1, S) > 0$ (resp. $\MOV(C_1, S) > 0$), and it concludes the proof. 

To generalize the proof for any $k, t > 2$, in the constructive model using the general non-increasing scoring function $f(\cdot)$, we set the preferences lists of voters such that the difference between the number of winners from our target party before and after a diffusion is exactly one, if and only if all of the nodes become active.

Consider the minimum index $j$ such that $1 \leq j-1 < j \leq |C|$ and $f(j-1) > f(j)$, i.e., $j$ is the minimum rank that has less score than rank $j-1$. 
For example for plurality scoring rule $j=2$, and for anti-plurality $j = |C|$. Note that by assumption there exists such a $j$.
Let us define $i:= j-1$, $a := f(i)$, and $b := f(j)$.

As in the case of $t=3$ and $k=2$, we provide a reduction from the decision version of the deterministic IM problem, to \CMEC and \CDW. 
Let $\mathcal{I}(G, B)$ be an instance for the decision IM problem, where $G = (V, E)$ and $B$, respectively, are the given directed deterministic graph and budget for the IM problem. 
We create an instance $\mathcal{I}'(G', B)$ for the general scoring rule \CMEC and \CDW, where $G'=(V \cup V', E \cup E')$. We consider the same budget, denoted by $B$ for both problems. In order to create $G'$ we do as follows.
\begin{itemize}
    \item For each $v \in V$ we add one more node in $V'$ and it has just one incoming edge from $v$, i.e., $\forall v \in V: v_1 \in V', (v, v_1) \in E'$. %By default $V''$ is an empty set.
    
    \item The weight of all edges is 1, i.e., $b_{uv} = 1$ for all $(u, v) \in E \cup E'$.
    
    \item We set the preferences of all nodes $v \in V$ and its new outgoing neighbor as follows:

    \begin{description}
        \item[C1.] If $j \leq k$.
        In this case, we add $k-i$ candidates to $C$, i.e., $C = C \cup \{c_1^{t+1}, \dots, c_{k-i}^{t+1}\}$; also, we define a new scoring function $f'(\cdot)$ as follows:
        \begin{align*}
            1 \leq z \leq k-i&: f'(z) = f(1), \\
            k-i < z \leq tk + k - i&: f'(z) = f(z - k + i).
        \end{align*}
        Note that the new scoring function holds the same conditions, i.e., it is still non-increasing, and there exists an index $k$ with $f'(k) = f(i)$ that has more score than $f'(k+1) = f(j)$. Then we set the preference list of the voters as follows.
        \begin{align*}
                v~ &: \can11 \succ \dots \succ c_{k-1}^1 \succ \overset{\overset{k}{\downarrow}}{\can12} \succ \overset{\overset{k+1}{\downarrow}}{\cank1} \succ \dots; \\
                v_1 &: \can12 \succ \can11 \succ \dots \succ c_{k-1}^1 \succ \overset{\overset{k+1}{\downarrow}}{\can13} \succ \overset{\overset{k + 2}{\downarrow}}{\cank1} \succ \dots.
            \end{align*}

            In this case, using the new scoring function $f'(\cdot)$, for $1 \leq z \leq k-1$ all the candidates $c_z^1 \in C_1, \can12 \in C_2$ have more score than $\cank1 \in C_1$; also, every other candidate ($c \notin C_1 \cup \{\can12\}$) has at most the same score of $\cank1$. Hence, $\F(C_1, \emptyset) = k-1, \F(C_2, \emptyset) = 1$, and for all $3 \leq r \leq t: \F(C_r, \emptyset) = 0$.

        \item[C2.] If $k < j < |C|$.
            \begin{align*}
                v~ &: \can11 \succ \dots \succ c_{k-1}^1 \succ \can22 \succ \dots \succ \cank2 \succ \can13 \succ \dots \succ \overset{\overset{i}{\downarrow}}{\can12} \succ \overset{\overset{j}{\downarrow}}{\cank1} \succ \dots; \\
                v_1 &: \can11 \succ \dots \succ c_{k-1}^1 \succ \can12 \succ \can22 \succ \dots \succ \cank2 \succ \dots \succ \overset{\overset{j}{\downarrow}}{\can13} \succ \overset{\overset{j + 1}{\downarrow}}{\cank1} \succ \dots.
            \end{align*}
        Also in this case, for $1 \leq z \leq k-1$, all of the candidates $c_z^1 \in C_1, \can12 \in C_2$ have more votes than $\cank1$; and every other candidate $c \in C\setminus (C_1 \cup \{c_1^2\})$ has at most the same score of $\cank1$. Therefore, $\F(C_1, \emptyset) = k-1, \F(C_2, \emptyset) = 1$, and for all $3 \leq r \leq t: \F(C_r, \emptyset) = 0$.

        \item[C3.] If $j = |C|$. In this case, we add an extra candidate called $c_{1}^{t+1}$ to the candidates ($C = C \cup \{c_1^{t+1} \}$), and set $f(tk+1) = f(tk)$. The preferences list of nodes is as follows:

        \begin{align*}
            v~ &: \can11 \succ \dots \succ c_{k-1}^1 \succ \can22 \succ \dots \succ \cank2 \succ \can13 \succ \dots \succ \overset{\overset{i}{\downarrow}}{\can12} \succ \overset{\overset{j}{\downarrow}}{\cank1} \succ c_1^{t+1}; \\
            v_1 &: \can11 \succ \can21 \succ \dots \succ c_{k-1}^1 \succ \can12 \succ \dots \succ \cank2 \succ \dots \succ  \overset{\overset{j}{\downarrow}}{c_1^{t+1}} \succ \overset{\overset{j + 1}{\downarrow}}{\cank1}.
        \end{align*}
        Since $f(tk+1) = f(tk)$, then the extra candidate will not be among winners before and after any diffusion. 
        We just add it to the candidates to be able to put the candidate $\cank1 \in C_1$ in position $j+1$, and this case will be the same as previous cases, i.e., $\F(C_1, \emptyset) = k-1, \F(C_2, \emptyset) = 1$, and for all $3 \leq r \leq t: \F(C_r, \emptyset) = 0$.
    \end{description}
\end{itemize}

By this reduction, before any diffusion, in all cases, $k-1$ winners will be from our target party $C_1$, and one winner from party $C_2$.

In $\mathcal{I}'(G', B)$, concerning all cases, all of the nodes $v \in V \cup V'$ become active if and only if all of the nodes $v \in V$ become active. Indeed, by definition, if $V \subseteq A_S$, then for each node $u \in V'$ there exists an incoming neighbor $v \in V$ s.t. $(v, u) \in E'$ and $b_{vu} = 1$, then if $v$ is active, also $u$ becomes active.

Suppose there exists an $\alpha-$approximation algorithm called \textit{$\alpha$-appAlg} for \CDW (respectively, \CMEC) and it returns $S \subseteq V \cup V'$ as a solution. 
We show that, by using the seed nodes $S$ returned by the algorithm \textit{$\alpha$-appAlg}, we can find the answer for the decision IM problem.
We will show that $\DOW(C_1, S) > 0$ (resp. $\MOV(C_1, S) > 0$) if and only if $S$ activates all of the nodes in $V \cup V'$, i.e.,  $A_S = V \cup V'$, where  $A_S$ is the set of activated nodes after $S$. That is, since $G$ is a subgraph of $(V \cup V', E \cup E')$, then  $\DOW(C_1, S) > 0$ (respectively, $\MOV(C_1, S) > 0$) if and only if the answer to the decision IM problem is YES.

In order to show if $\DOW(C_1, S) > 0$ (resp. $\MOV(C_1, S) > 0$), then $S$ activates all of the nodes in $V \cup V'$, i.e.,  $A_S = V \cup V'$, first we observe that by $S' \subset V \cup V'$ returned by \textit{$\alpha$-appAlg}, we can find a new seed nodes $S \subseteq V$ s.t. $\DOW(C_1, S') = \DOW(C_1, S)$ (also $\MOV(C_1, S') = \MOV(C_1, S)$). 
To do that, if there exists a node $v$ in $S \cap V'$, then we can replace it with the node $u \in V$ s.t. $(u, v) \in E'$; by this substitution, our target candidates get at least the same score as before, and $\DOW(C_1, S') = \DOW(C_1, S)$ (also $\MOV(C_1, S') = \MOV(C_1, S)$).
Therefore, from now on, we assume that \textit{$\alpha$-appAlg} returns the seed nodes $S$, s.t. $S \in V$.

Now we show that if $\DOW(C_1, S) > 0$ (resp. $\MOV(C_1, S) > 0$), then $S$ activates all of the nodes in $V \cup V'$, i.e.,  $A_S = V \cup V'$. 
By contradiction, assume $S$ is a solution for the problem s.t. $\DOW(C_1, S) > 0$ (resp. $\MOV(C_1, S) > 0$), and $S$ does not activate all nodes of $V \cup V'$, i.e., there exists a node $v$ s.t. $v \in V \setminus A_S$.
In this case, for the score of candidates, for all cases,  according to the preferences lists, the candidate $\cank1$ will have at least $2(a-b)$ score less than candidates  $\can11, \dots, c_{k-1}^1, \can12$, which means $\F(C_1, S) = k-1, \F(C_2, S) = 1$. According to the definition of $\DOW$ we have $\DOW(C_1, S) = 0$ (also, $\MOV(C_1, S) = 0$), which is a contradiction.
On the other hand, to show that if $S$ activates all the nodes $V \cup V'$, then $\DOW(C_1, S) > 0$ (resp. $\MOV(C_1, S) > 0$), note that if all nodes $V \cup V'$ become active, the score of all candidates $c \in C_1$ will be more than or equal to all other candidates. 
Then by using the tie breaking rule, $\F(C_1, S) = k$, and for $2 \leq n \leq t: \F(C_n, S) = 0$, which yields $\DOW(C_1, S) > 0$ (resp. $\MOV(C_1, S) > 0$).
Then, on one side, if $S$ is an answer to the \DMEC or \DDW, it activates all $2n$ nodes in $V \cup V'$.
On the other side, since $G$ is subgraph of $(V \cup V', E \cup E')$, and $S \subseteq V$, then $S$ activates all nodes in $G$.
%\qed
\end{proof}

%% file: trunk/hardness/destructive.tex
%\subsection{Destructive Election Control}
%\label{subsec:destructivePlurality}
\paragraph{Destructive Election Control}
%In the destructive case, the adversary tries to reduce the number of winners from the target party.
The following theorem shows the hardness of approximation of the destructive case. The proof is similar to that of Theorem~\ref{theorem:alphaAppNpHardness}. %, and hence we provide a proof sketch and is omitted due to space limits.
Note that if we consider maximizing $\DOWD$, the destructive case can be reduced to the constructive model. We cannot apply the same reduction to the problem of maximizing $\MOVD$ as the opponent party with the highest number of winners (i.e., $\Cmob, \Cmoa{S}$) may be different from that of the constructive case.
\begin{thrm}
    It is $\NP$-hard to approximate \DMEC and \DDW within any factor $\alpha > 0$.
\end{thrm}

\begin{proof}%[Proof Sketch]
Consider the proof in Theorem~\ref{theorem:alphaAppNpHardness} for $t = 3, k = 2$, where we set the preferences list of voters so that activating a node will increase the score of our target candidates, and all of our target candidates will win the election if and only if all nodes become active after $S$.
We can apply the same idea for the destructive case, except that all of our target candidates win the election before any diffusion, and all of them will lose if and only if all the nodes $v \in V$ become active after $S$.

To achieve that, we can assign the preferences list of the voters so that the score of our target candidates $c \in C_1$ is equal and more than the score of all other specific candidates (e.g. $c' \in C_2$) before any diffusion; but after $S$, the score of all candidates $c \in C_1$ become less than the score of the $k$ specific candidates if and only if all nodes $v \in V$ are activated.
Note that in the destructive case, we need $k$ isolated nodes more than the constructive case since, in the tie-condition, our target candidates $c \in C_1$ will be preferred to other candidates. Then by using the $k$ more isolated nodes, we make sure that all candidates $c \in C_1$ have less score than other specific candidates.
\end{proof}

%% file: trunk/SPVoting.tex
\section{Straight-Party Voting}
\label{sec:straightParty}
Since all the variants of the multi-winner election control problems considered so far are $\NP$-hard to be approximated within any factor, we now consider a relaxation of the problem in which, instead of focusing on the number of elected target candidates, we focus on the overall number of votes obtained by the target party. 
The rationale is that, even if a party is not able to (approximately) maximize the number of its winning candidates because it is computationally unattainable, it may want to maximize the overall number of votes, in the hope that these are not too spread among the candidates and still leads to a large number of seats in the parliament.

Moreover, the voting system that we obtain by the relaxation is used in some real parliamentary elections~\cite{ruhl2019party}, and is called of \textit{Straight-party voting (\SPV)} or \textit{straight-ticket voting}~\cite{campbell2009straight,kritzer2015roll}. %\SPV is the practice of voting for every candidate of a political party in an election. 
\SPV was used very much until around the 1960s and 1970s in the United States. 
After that, the United States has declined \SPV among the general voting; nevertheless, strong partisans are still voting according to \SPV.
Interestingly, the first time that every state voting for a Democrat for Senate also voted Democratic for president (and the same stability for Republicans) was the 2016 elections of the United states~\cite{hershey2017party}.

Note that in this model, if we consider that the controller targets a single candidate instead of a party and the preference lists are over candidates, then we can easily reduce the problem to the single-winner case. The same holds if the controller targets a party and the preference lists are over parties.
Therefore, we assume that voters have preference lists over the candidates, but since the voting system is \SPV and voters have to vote for a party, then they will cast a vote for each party based on the position of the candidates of the party in their preferences list, e.g., if the preferences list of a node $v \in V$ is $\can11 \succ \can12 \succ \can22 \succ \can21$, then the scores of $v$ for party $C_1$ will be $f(1)+f(4)$, and $f(2) + f(3)$ for party $C_2$.

%if voters have preferences lists over the parties (and not the candidates), the problem can be reduced to the single-winner election control easily.
\iffalse
\todo{G: I would remove the following paragraph because it is not really needed for the model and it is too informal.}
In the relaxed model, the parties will select their candidates after the election, and the number of seats for each party is proportional to the number of votes gained by that party. 
For example, if party $C_1$ gets ten percent of the votes at the end of the election, then ten percent of the winning candidates are the members of party $C_1$. 
If some parties are running for a seat with less than the required share, the party that has the maximum remaining share will win. For example, assume $t = 3, k = 10$, then parties should gain ten percent of the total score for each seat. Suppose the share of each party after $S$ is $\F(C_1, S) = 45\%, \F(C_2, S) = 37\%, \F(C_3, S) = 18\%$; in this situation $C_1, C_2, C_3$ will have 4, 4, 2 winner candidates, respectively. As a tie-breaking rule, if the remaining share is equal for some parties, party $C_i$ is preferred to $C_j$ if $i < j$.
In this case, we care about maximizing the share of our target party concerning the most voted opponent party.
We remark that this model is a relaxed version of that in Section~\ref{sec:multiWinnerElectionControl}, where we take into account the number of votes received by a party instead of the number of winning candidates.
\fi

Let us define 
$\F_{\SPVF}(C_i, \emptyset)$ and $\F_{\SPVF}(C_i, S)$ as sum of the scores obtained by party $C_i$ in \SPV before and after $S$, respectively, as follows.
\begin{align*}
\F_{\SPVF}(C_i, \emptyset) = \sum_{v \in V} \sum_{c \in C_i} f(\pi_v(c)), \\
\F_{\SPVF}(C_i, S) = \mathbb{E} \big[ \sum_{v \in V} \sum_{c \in C_i} f(\tilde{\pi}_v(c)) \big].
\end{align*}

As in the previous case, we denote by $\Cmob$ and $\Cmoa{S}$ the most voted opponents of $C_1$ before and after $S$, respectively. We define \MOV and \MOVD for \SPV as
\begin{align*}
    \MOV^{\SPVF}(C_1, S) &= \F_{\SPVF}(C_1, S) - \F_{\SPVF}(\Cmoa{S}, S) - %\\ & \quad
    \big( \F_{\SPVF}(C_1, \emptyset) - \F_{\SPVF}(\Cmob, \emptyset) \big),\\
    \MOVD^{\SPVF}(C_1, S) &= \F_{\SPVF}(C_1, \emptyset) - \F_{\SPVF}(\Cmob, \emptyset) -  %\\ & \quad
    \big( \F_{\SPVF}(C_1, S) - \F_{\SPVF}(\Cmoa{S}, S) \big).
\end{align*}
Also, we define \textit{Difference of votes} for constructive (\DOV) and destructive (\DOVD) as
\begin{align*}
    \DOV^{\SPVF}(C_1, S) &= \F_{\SPVF}(C_1, S) - \F_{\SPVF}(C_1, \emptyset),\\
    \DOVD^{\SPVF}(C_1, S) &= \F_{\SPVF}(C_1, \emptyset) - \F_{\SPVF}(C_1, S).
\end{align*}

%In the following theorem, we prove that it is $\NP$-hard to approximate the 

\begin{thrm}
\label{theorem:SPVApproximationHardness}
It is $\NP$-hard to approximate the $\MOV^{\SPVF}, \DOV^{\SPVF}, \MOVD^{\SPVF}, \mbox{ and }\DOVD^{\SPVF}$ to within better than a factor $1-\frac{1}{e}$.
\end{thrm}

\begin{proof}
Consider an instance of the deterministic IM problem $\I(G, B)$ using the Independent Cascade (IC) model. 
We reduce it to an instance of \SPV problem $\I'(G, B)$, using the same graph $G = (V, E)$, and budget $B$ for both problems.
As in Theorem~\ref{theorem:alphaAppNpHardness}, we define $j$ as the minimum index such that $1 \leq j-1 < j \leq |C|$ and $f(j-1) > f(j)$, i.e., $j$ is the minimum rank that has less score than rank $j-1$, $2 \leq j \leq |C| = tk$. 
For simplicity, we define $i:=j-1$, $a := f(i)$, and $b := f(j)$.
We analyze two distinct cases and we set the preferences list of each node $v \in V$ as follows.
\begin{description}
    \item[C1.] If $j \leq k$:
    $$v: \can12  \succ \can22 \succ \dots \succ \overset{\overset{i}{\downarrow}}{c_i^2} \succ \overset{\overset{j}{\downarrow}}{\can11} \succ \can21 \succ \dots \succ \cank1 \succ c_j^2 \succ \dots$$ 
    In this case,  
    $\F_{\SPVF}(C_1, \emptyset) = |V|\sum_{l = j}^{j+k-1} f(l)$.

    \item[C2.] If $j > k$: 
    $$v: \can11 \succ \can21 \succ \dots \succ c_{k-1}^1 \succ \can12 \succ \can22 \succ \dots \succ \overset{\overset{j}{\downarrow}}{\cank1} \succ \dots$$ 
    In this case,
    $\F_{\SPVF}(C_1, \emptyset) = (a(k-1) + b)|V|$.
\end{description}

Let  $S$ be a set of seed nodes selected selected as a solution for $\I(G, B)$, and $A_S$ is the set of activated nodes after $S$. If we use the same seed set as a solution for $\I'(G, B)$, we have:
\begin{description}
    \item[C1.]
    \begin{align*}
        \F_{\SPVF}(C_1, S) &= |A_S|\sum_{l = i}^{i+k-1} f(l) + |V \setminus A_S|\sum_{l = j}^{j+k-1} f(l) \\
        \DOV^{\SPVF}(C_1, S) &= |A_S|\sum_{l = i}^{i+k-1} f(l) + |V \setminus A_S|\sum_{l = j}^{j+k-1} f(l) - |V|\sum_{l = j}^{j+k-1} f(l) \\
        &= |A_S|\Big(f(i) - f(i+k) \Big).
    \end{align*}
    
    \item[C2.] 
    \begin{align*}
        \F_{\SPVF}(C_1, S) &= ak|A_S| + |V \setminus A_S|\Big(a(k-1) + b\Big)\\
        \DOV^{\SPVF}(C_1, S) &= ak|A_S| + |V \setminus A_S|\Big(a(k-1) + b\Big) - \Big(a(k-1) + b\Big)|V|\\
        &= |A_S|(a - b).
    \end{align*}
\end{description}

On the other hand, if $S$ is a set of seeds selected as a solution for $\I'(G, B)$, and $A_S$ is the set of active nodes after $S$, and we use the same seed set as a solution for $\I(G, B)$, then we activate a number of nodes $|A_S|$ such that
\begin{description}
    \item[C1.] 
    \begin{align}
    \label{C1:dov}
        \DOV^{\SPVF}(C_1, S) = |A_S|\Big(f(i) - f(i+k) \Big).
    \end{align}
    \item[C2.] 
    \begin{align}
    \label{C2:dov}
        \DOV^{\SPVF}(C_1, S) = |A_S|(a - b).
    \end{align}
\end{description}
Overall, we proved that $|A_S|$ nodes will be activated by seed set $S$ if and only if using $S$ for maximizing $\DOV^{\SPVF}$, we get~\eqref{C1:dov} or~\eqref{C2:dov}. 
Since it is hard to approximate IM problem within a factor greater than $1-\frac{1}{e}$ then it yields the same hardness of approximation for maximizing $\DOV^{\SPVF}$.

%\todo{MAM: I have to finalize this part.}
Using the same reduction, we can prove the hardness of approximation for $\MOV^{\SPVF}$.
To see that, note that if we have seed set $S$ and $A_S$ are the activated nodes after $S$, then 
\begin{description}
    \item[C1.] In this case, party $C_2$ is the most scored opponent before and after $S$; then it yields
    \begin{align*}
        &\MOV^{\SPVF}(C_1, S) = \F_{\SPVF}(C_1, S) - \F_{\SPVF}(C_2, S) -  \\ & \quad \quad \quad \quad \quad \quad \quad \quad
        \Big( \F_{\SPVF}(C_1, \emptyset) - \F_{\SPVF}(C_2, \emptyset) \Big)\\
        &= |A_S|\sum_{l = i}^{i+k-1} f(l) + |V \setminus A_S|\sum_{l = j}^{j+k-1} f(l) - |V| \Big(\sum_{l=1}^{l = i-1}f(l) + \sum_{l = i+k}^{2k} f(l)\Big) - \\ & \quad \Big(|V|\sum_{l = j}^{j+k-1} f(l) - |V| (\sum_{l=1}^{l = i}f(l) + \sum_{l = j+k}^{2k} f(l)) \Big) \\
        & = |V|\Big(f(i) - f(i+k)\Big) + |A_S|\Big(f(i) - f(i+k)\Big).
    \end{align*}
    Note that in this case, since the party with the most scored before and after $S$ is $C_2$, then whatever $C_2$ looses, our target party $C_1$ will gain, and the score of other parties is the same as before diffusion. Then $\MOV^{\SPVF}(C_1, S) = 2\DOV^{\SPVF}(C_1, S)$.
    
    \item[C2.] In this case, we analyze two different sub-cases:
    \begin{description}
        \item[C2.1.] If $j \leq 2k$; then the party with most score before and after any diffusion is $C_2$. Similar to case \textbf{C1}, $\MOV^{\SPVF}(C_1, S) = 2\DOV^{\SPVF}(C_1, S)$.
        
        \item[C2.2.] If $2k < j$; the party with the most score before and after $S$ is still $C_2$, but the score of $C_2$ will not reduce after any diffusion, i.e., 
        \begin{align*}
            &\MOV^{\SPVF}(C_1, S) = \\ & \quad \F_{\SPVF}(C_1, S) - \F_{\SPVF}(\Cmoa{S}, S) - \left( \F_{\SPVF}(C_1, \emptyset) - \F_{\SPVF}(\Cmob, \emptyset) \right) \\
        &= ak|V| - ak|V| - \Big((a(k-1) + b)|V|) - ak|V| \Big) = (a - b)|V|.
        \end{align*}
    \end{description}
\end{description}
Since this holds for both directions, then we get the same harness of approximation factor for $\MOV^{\SPVF}$ too.

Following the same approach, we can prove the same hardness of approximation for the problems of maximizing $\MOVD^{\SPVF}$ and $\DOVD^{\SPVF}$. In the constructive case, we arranged the candidates so that each activated node will increase the score of the target party after $S$. In destructive case, we have to set the preferences list of voters so that each active node decrease the score of our target party.
%The proofs for destructive case are omitted as they are very similar to the constructive model.
%\qed
\end{proof}

We now give an approximation algorithm for the problems of maximizing $\DOV^{\SPVF}$ and $\DOVD^{\SPVF}$ that is based on a reduction to the node-weighted version of the IM problem.
We construct an instance of this problem where the weight to each node $v \in V$, which is equal to the increase in the score of $C_1$ when $v$ becomes active. The node-weighted IM problem can be approximated by a factor of $1-\frac{1}{e}-\epsilon$, for any $\epsilon>0$, by using the standard greedy algorithm~\cite{kempe2015maximizing}.

\begin{thrm}
\label{lemma:maximizingDOV}
There exists an algorithm that approximates $\DOV^{\SPVF}$ and $\DOVD^{\SPVF}$ within a factor $(1-\frac{1}{e})-\epsilon$ from the optimum, for any $\epsilon>0$.
\end{thrm}
\begin{proof}
We first consider the constructive case, i.e., $\DOV^{\SPVF}$. 
Let us define $\bar{C}_1^v \subseteq C_1$ as a set of candidates in our target party whose rank is decreased if $v$ become active; in other words, $\bar{C}_1^v = \{c \in C_1: \exists c' \in C \setminus C_1, \pi_v(c') < \pi_v(c)\}$.
In this case, a node $v \in V$ can increase the score of $C_1$ by $\sum_{c \in \bar{C}_1^v} %C_1: \exists c' \in C \setminus C_1, \pi_v(c') < \pi_v(c)} 
f(\pi_v(c)-1) - f(\pi_v(c))$.\footnote{We assume function $f(\cdot)$ is defined in such away that $f(i-1) - f(i)$, for $i=2,\ldots,m$, does not depend exponentially on the graph size (e.g. it is a constant). The influence maximization problem with arbitrary node-weights is still an open problem~\cite{kempe2015maximizing}.}
Given an instance of $\mathcal{I}(G, B)$ of the $\DOV^{\SPVF}$ maximization problem, we define an instance $\mathcal{I'}(G, B, w)$ of the node-weighted IM problem, where $w$ is a node-weight function defined as $w(v) = \sum_{c \in \bar{C}_1^v}\left( f(\pi_v(c)-1) - f(\pi_v(c))\right),$ 
% \[
% w(v) = \sum_{c \in \bar{C}_1^v}\left( f(\pi_v(c)-1) - f(\pi_v(c))\right),
% \]
for all $v\in V$. %, where $\bar{C}_1^v = \{c \in C_1: \exists c' \in C \setminus C_1, \pi_v(c') < \pi_v(c)\}$.
Given a set $S$ of nodes, we denote by $\sigma(S)$ the expected weight of active nodes in $G$, when the diffusion starts from $S$. We will show that $\DOV^{\SPVF}(C_1, S) = \sigma(S)$ for any set $S\subseteq V$, since the standard greedy algorithm guarantees an approximation factor of $1-\frac{1}{e}-\epsilon$, for the node-weight IM problem, for any $\epsilon>0$, this shows the statement.

Given a set $S$, $\sigma(S)$ can be computed as follows:
\begin{align*}
    \sigma(S) &= \mathbb{E}_{A_S}\left[\sum_{v\in A_S} w(v)\right] = \sum_{A_S\subseteq V}\sum_{v\in A_S}w(v)\mathbb{P}(A_S),
\end{align*}
where $\mathbb{P}(A_S)$ is the probability that $A_S\subseteq V$ is the set of nodes activated by $S$. 

By definition $\DOV^{\SPVF}(C_1, S) = \F_{\SPVF}(C_1, S) - \F_{\SPVF}(C_1, \emptyset)$, where $\F_{\SPVF}(C_1, \emptyset) = \sum_{v \in V} \sum_{c \in C_1}f(\pi_v(c))$ and
\begin{align*}
    \F_{\SPVF}(C_1, S) &= \mathbb{E}_{A_S} \left[ \sum_{v \in V} \sum_{c \in C_1} f(\tilde{\pi}_v(c)) \right] = \sum_{v\in V}\sum_{c\in C_1}\mathbb{E}_{A_S} \left[f(\tilde{\pi}_v(c))\right] \\
    &= \sum_{v\in V}\left( \sum_{c\in \bar{C}_1^v}\mathbb{E}_{A_S} \left[f(\tilde{\pi}_v(c))\right] + \sum_{c\in C_1\setminus \bar{C}_1^v}f(\pi_v(c)) \right),
\end{align*}
where the last equality is due to the fact that, a node $v$ doesn't change the positions of candidates in $C_1\setminus \bar{C}_1^v$.
Let us focus on the first term of the above formula,
\begin{align*}
   &\sum_{v\in V} \sum_{c\in \bar{C}_1^v}\mathbb{E}_{A_S} \left[f(\tilde{\pi}_v(c))\right]\\  
   &\quad\quad\quad\quad\quad\quad = \sum_{v\in V} \sum_{c\in \bar{C}_1^v} \sum_{A_S\subseteq V} \left( f(\pi_v(c)-1)\mathds{1}_{v\in A_s} +f(\pi_v(c))\mathds{1}_{v\not\in A_s} \right)\mathbb{P}(A_S)\\
   &\quad\quad\quad\quad\quad\quad = \sum_{A_S\subseteq V}\left(\sum_{v\in A_S}\sum_{c\in \bar{C}_1^v}
   f(\pi_v(c)-1) + \sum_{v\not\in A_S}\sum_{c\in \bar{C}_1^v}f(\pi_v(c))
   \right)\mathbb{P}(A_S).
\end{align*}
It follows that
\begin{align*}
    \DOV^{\SPVF}(C_1, S) &= \sum_{v \in V} \sum_{c\in \bar{C}_1^v}\mathbb{E}_{A_S} \left[f(\tilde{\pi}_v(c)) - f(\pi_v(c))\right]\\
    & \sum_{A_S\subseteq V}\sum_{v\in A_S}\sum_{c\in \bar{C}_1^v}
   (f(\pi_v(c)-1) - f(\pi_v(c)))\mathbb{P}(A_S) = \sigma(S),
\end{align*}
since the term related to candidates in $C_1\setminus \bar{C}_1^v$ and to nodes not in $A_S$ are canceled out.

The destructive case is similar to the constructive one except that a node $v \in V$ can decrease the score of $C_1$ by $\sum_{c \in C_1: \exists c' \in C \setminus C_1, \pi_v(c') > \pi_v(c)} f(\pi_v(c)) - f(\pi_v(c)+1)$.
Therefore the same approach, where the weights are set to the above value, yields the same approximation factor for $\DOVD^{\SPVF}$.
%\qed
\end{proof}

In the following theorems, we show that using Theorem~\ref{lemma:maximizingDOV}, we get a constant approximation factor for the problem of maximizing MoV. Specifically, we show that by maximizing $\DOV^{\SPVF}$ we get an extra $1/3$ approximation factor for the problem of maximizing $\MOV^{\SPVF}$. For the destructive case, the extra approximation factor is $1/2$.
It follows that, by using the greedy algorithm for maximizing $\DOV^{\SPVF}$ and $\DOVD^{\SPVF}$, we obtain approximation factors of $\frac{1}{3}(1-\frac{1}{e})-\epsilon$ and $\frac{1}{2}(1-\frac{1}{e})-\epsilon$, for any $\epsilon>0$, of the maximum $\MOV^{\SPVF}$ and $\MOVD^{\SPVF}$, respectively.

\begin{thrm}
There exists an algorithm that approximates $\MOV^{\SPVF}$  within a factor $\frac{1}{3}(1-\frac{1}{e})-\epsilon$ from the optimum, for any $\epsilon>0$.
%
%By maximizing $\DOV^{\SPVF}$ we obtain  $\frac{1}{3}(1-\frac{1}{e})$-approximation factor for $\MOV^{\SPVF}$.
\end{thrm}
\begin{proof}
Let $S$ and $S^*$ be the solution returned by the greedy algorithm for $\DOV^{\SPVF}$ maximization and a solution that maximizes $\MOV^{\SPVF}$, respectively. For each party $C_i\neq C_1$, we denote by $\DOVN(C_i, S)$ the score lost by $C_i$ after $S$, that is $\DOVN(C_i, S) = \F(C_i,\emptyset) - \F(C_i,S)\geq 0$. Let $\alpha_\epsilon := (1-\frac{1}{e})-\epsilon$. Since $S$ is a factor $\alpha_\epsilon$ from the optimum $\DOV^{\SPVF}$, the following holds.
\begin{align*}
    &\MOV^{\SPVF}(C_1, S) = \F_{\SPVF}(C_1, S) - \F_{\SPVF}(\Cmoa{S}, S) - \Big( \F_{\SPVF}(C_1, \emptyset) - \F_{\SPVF}(\Cmob, \emptyset) \Big)\\
    &\quad= \DOV^{\SPVF}(C_1, S) + \DOVN(\Cmoa{S}, S) - \F_{\SPVF}(\Cmoa{S}, \emptyset) + \F_{\SPVF}(\Cmob, \emptyset)\\
    &\quad\geq \alpha_\epsilon\DOV^{\SPVF}(C_1, S^*) - \F_{\SPVF}(\Cmoa{S}, \emptyset) + \F_{\SPVF}(\Cmob, \emptyset)\\
    &\quad\stackrel{(a)}{\geq} \frac{1}{3}\alpha_\epsilon\left[\DOV^{\SPVF}(C_1, S^*) + \DOVN(\Cmoa{S^*}, S^*) + \DOVN(\Cmoa{S}, S^*) \right] -\\ & \quad \quad \F_{\SPVF}(\Cmoa{S}, \emptyset) + \F_{\SPVF}(\Cmob, \emptyset)\\
    &\quad\stackrel{(b)}{\geq} \frac{1}{3}\alpha_\epsilon\Big[\DOV^{\SPVF}(C_1, S^*) + \DOVN(\Cmoa{S^*}, S^*) + \DOVN(\Cmoa{S}, S^*) -\\ & \quad \quad \F_{\SPVF}(\Cmoa{S}, \emptyset) + \F_{\SPVF}(\Cmob, \emptyset) + \F_{\SPVF}(\Cmoa{S^*}, \emptyset) - \F_{\SPVF}(\Cmoa{S^*}, \emptyset) \Big]\\
    &\quad= \frac{1}{3}\alpha_\epsilon \Big[ \MOV^{\SPVF}(C_1, S^*) + \DOVN(\Cmoa{S}, S^*) + \F_{\SPVF}(\Cmoa{S^*}, \emptyset) - \F_{\SPVF}(\Cmoa{S}, \emptyset) \Big]\\
    &\quad\stackrel{(c)}{\geq} \frac{1}{3}\alpha_\epsilon\MOV^{\SPVF}(C_1, S^*) \geq \left(\frac{1}{3}\left(1-\frac{1}{e}\right)-\epsilon \right)\MOV^{\SPVF}(C_1, S^*),
\end{align*}
for any $\epsilon>0$.
Inequality $(a)$ holds because, by definition, the score lost by $\Cmoa{S}$ and $\Cmoa{S^*}$ will be added to the score of $C_1$. Inequality $(b)$ holds since, by definition of $\Cmob$, $\F_{\SPVF}(\Cmob, \emptyset) \geq \F_{\SPVF}(\Cmoa{S}, \emptyset)$ and then $- \F_{\SPVF}(\Cmoa{S}, \emptyset) + \F_{\SPVF}(\Cmob, \emptyset)\geq 0$.
Inequality $(c)$ holds because 
\begin{align*}
    &\DOVN(\Cmoa{S^*}, S^*) - \DOVN(\Cmoa{S}, S^*) \\
    &\quad\quad= \F_{\SPVF}(\Cmoa{S^*}, \emptyset) - \F_{\SPVF}(\Cmoa{S^*}, S^*) - \F_{\SPVF}(\Cmoa{S}, \emptyset) + \F_{\SPVF}(\Cmoa{S}, S^*) \\
    &\quad\quad\stackrel{(d)}{\leq} \F_{\SPVF}(\Cmoa{S^*}, \emptyset) - \F_{\SPVF}(\Cmoa{S}, \emptyset),
  \end{align*}
which implies that 
\begin{align*}
\DOVN(\Cmoa{S}, S^*) + \F_{\SPVF}(\Cmoa{S^*}, \emptyset) - \F_{\SPVF}(\Cmoa{S}, \emptyset)  \geq \DOVN(\Cmoa{S^*}, S^*) \geq 0.
\end{align*}
Inequality  $(d)$ holds since, by definition of $\Cmoa{S^*}$, $\F(\Cmoa{S}, S^*) \leq \F(\Cmoa{S^*}, S^*)$.
%\qed
\end{proof}

\begin{thrm}
There exists an algorithm that approximates $\MOVD^{\SPVF}$  within a factor $\frac{1}{2}(1-\frac{1}{e})-\epsilon$ from the optimum, for any $\epsilon>0$.
%By maximizing $\DOVD^{\SPVF}$ we obtain  $\frac{1}{2}(1-\frac{1}{e})$-approximation factor for $\MOVD^{\SPVF}$.
\end{thrm}

\begin{proof}

Let $S$ and $S^*$ be the solution returned by the greedy algorithm for $\DOVD^{\SPVF}$ maximization and a solution that maximizes $\MOVD^{\SPVF}$, respectively. For each party $C_i\neq C_1$, we denote by $\DOVP(C_i, S)$ the score gained by $C_i$ after $S$, that is $\DOVP(C_i, S) = \F(C_i,S) - \F(C_i,\emptyset) \geq 0$. Let $\alpha_\epsilon := (1-\frac{1}{e})-\epsilon$. Since $S$ is a factor $\alpha_\epsilon$ from the optimum $\DOV^{\SPVF}$, the following holds.
\begin{align*}
    &\MOVD^{\SPVF}(C_1, S) = \F_{\SPVF}(C_1, \emptyset) - \F_{\SPVF}(\Cmob, \emptyset) - %\\ & \quad 
    \left( \F_{\SPVF}(C_1, S) - \F_{\SPVF}(\Cmoa{S}, S) \right)\\
    &\quad= \DOVD^{\SPVF}(C_1, S) + \DOVP(\Cmoa{S}, S) + \F_{\SPVF}(\Cmoa{S}, \emptyset) - \F_{\SPVF}(\Cmob, \emptyset)\\
    &\quad\geq \alpha_\epsilon\DOVD^{\SPVF}(C_1, S^*) + \DOVP(\Cmoa{S}, S) + \F_{\SPVF}(\Cmoa{S}, \emptyset) - \F_{\SPVF}(\Cmob, \emptyset)\\
    &\quad\stackrel{(a)}{\geq} \frac{1}{2}\alpha_\epsilon\left[\DOVD^{\SPVF}(C_1, S^*) + \DOVP(\Cmoa{S^*}, S^*) \right] + \DOVP(\Cmoa{S}, S) + \\ & \quad \quad \F_{\SPVF}(\Cmoa{S}, \emptyset) - \F_{\SPVF}(\Cmob, \emptyset)\\
    &\quad\stackrel{(b)}{\geq} \frac{1}{2}\alpha_\epsilon\Big[\DOVD^{\SPVF}(C_1, S^*) + \DOVP(\Cmoa{S^*}, S^*) + \DOVP(\Cmoa{S}, S) + \\ & \quad \quad \F_{\SPVF}(\Cmoa{S}, \emptyset) - \F_{\SPVF}(\Cmob, \emptyset) + \F_{\SPVF}(\Cmoa{S^*}, \emptyset) - \F_{\SPVF}(\Cmoa{S^*}, \emptyset)\Big]\\
    &\quad= \frac{1}{2}\alpha_\epsilon \Big[ \MOVD(C_1, S^*) + \DOVP(\Cmoa{S}, S) - \F_{\SPVF}(\Cmoa{S^*}, \emptyset) + \F_{\SPVF}(\Cmoa{S}, \emptyset) \Big]\\
    &\quad\stackrel{(c)}{\geq} \frac{1}{2}\alpha_\epsilon\MOVD(C_1, S^*) \geq \left(\frac{1}{2}\left(1-\frac{1}{e}\right)-\epsilon\right)\MOVD(C_1, S^*),
\end{align*}
for any $\epsilon>0$. Inequality $(a)$ holds because $\Cmoa{S^*}$ can gain at most all of the scores lost by $C_1$.
Inequality $(b)$ holds since we have
\begin{align*}
    & \DOVP(\Cmoa{S}, S) + \F_{\SPVF}(\Cmoa{S}, \emptyset) - \F_{\SPVF}(\Cmob, \emptyset) \\
    &\quad\quad= \F_{\SPVF}(\Cmoa{S}, S) - \F_{\SPVF}(\Cmoa{S}, \emptyset) + \F_{\SPVF}(\Cmoa{S}, \emptyset) - \F_{\SPVF}(\Cmob, \emptyset) \\
    &\quad\quad= \F_{\SPVF}(\Cmoa{S}, S)  - \F_{\SPVF}(\Cmob, \emptyset),
\end{align*}
and, by definition of $\Cmoa{S}$, $\F_{\SPVF}(\Cmoa{S}, S) \geq\F_{\SPVF}(\Cmob, S)\geq \F_{\SPVF}(\Cmob, \emptyset)$.
Inequality~$(c)$ holds because 
\begin{align*}
     &\DOVP(\Cmoa{S}, S) - \F_{\SPVF}(\Cmoa{S^*}, \emptyset) + \F_{\SPVF}(\Cmoa{S}, \emptyset) \\
     &\quad\quad= \F_{\SPVF}(\Cmoa{S}, S) - \F_{\SPVF}(\Cmoa{S}, \emptyset) - \F_{\SPVF}(\Cmoa{S^*}, \emptyset) + \F_{\SPVF}(\Cmoa{S}, \emptyset)\\
     &\quad\quad= \F_{\SPVF}(\Cmoa{S}, S) - \F_{\SPVF}(\Cmoa{S^*}, \emptyset)
\end{align*}
and, by definition of $\Cmoa{S}$, $\F_{\SPVF}(\Cmoa{S}, S) \geq \F_{\SPVF}(\Cmoa{S^*}, S) \geq \F_{\SPVF}(\Cmoa{S^*}, \emptyset)$.
%\qed
\end{proof}

%% file: trunk/conclusions.tex
\section{Conclusions and Future Work}
\label{sec:conclusion}
%\todo{revise this section.}
Controlling elections through social networks is a significant issue in modern society. 
Political campaigns are using social networks as effective tools to influence voters in real-life elections.
In this paper, we formalized the multi-winner election control problem through social influence.
We proved that finding an approximation to the maximum margin of victory or difference of winners, for both constructive and destructive cases, is $\NP$-hard for any approximation factor. We relaxed the problem to a variation of straight-party voting and showed that this case is approximable within a constant factor in both constructive and destructive cases.
To our knowledge, these are the first results on multi-winner election control via social influence.

The results in this paper open several research directions. 
We plan to study the problem in which the adversary can spread a different (constructive/destructive) message for each candidate, using different seed nodes. In these cases, a good strategy could be that of sending a message regarding a third party (different from the target one and the most voted opponent), and our results cannot be easily extended.